\documentclass{ws-ijmpc}
\usepackage[T1]{fontenc}
\usepackage[latin9]{inputenc}
\usepackage{array}
\usepackage{amssymb}
\usepackage{graphicx}
\usepackage[numbers]{natbib}

\makeatletter

\providecommand{\tabularnewline}{\\}

\usepackage{amsmath}

\usepackage{shuffle}
\usepackage{enumerate}

\makeatother

\begin{document}
\title{\uppercase{On Basic Properties of Jumping Finite Automata}\footnote{Research supported by the Czech Science Foundation grant GA14-10799S and the GAUK grant No. 52215.}}

\author{\uppercase{Vojt\v{e}ch Vorel}}
\address{Department of Theoretical Computer Science and Mathematical Logic,\\
Charles University,\\ Malostransk\'{e} n\'{a}m. 25, Prague, Czech Republic\\ vorel@ktiml.mff.cuni.cz}

\maketitle

\begin{abstract} We complete the initial study of jumping finite automata, which was started in a former article of Meduna and Zemek \citep{athMED1}. The open questions about basic closure properties are solved. Besides this, we correct erroneous results presented in the article. Finally, we point out important relations between jumping finite automata and some other models studied in the literature. \end{abstract}


\section{Introduction}

In 2012, Meduna and Zemek \citep{athMED1} introduced \emph{general
jumping finite automata} as a model of discontinuous information processing.
A general jumping finite automaton (GJFA) is described by a finite
set $Q$ of states, a finite alphabet $\Sigma$, a finite set $R$
of \emph{rules }from $Q\times\Sigma^{*}\times Q$, an initial state
$s\in Q$, and a set $F\subseteq Q$ of final states. In a step of
the computation, the automaton switches from a state $q$ to a state
$r$ using a rule $\left(q,v,r\right)\in R$ and deletes a factor
equal to $v$ from any part of the input word. The choices of the
rule used and of the factor deleted are made nondeterministically.
A word can be accepted only if there is a computation resulting in
the empty word.

There is an infinite hierarchy of GJFA according to the maximum length
of factor deleted in a single step -- a GJFA is of \emph{degree} $n$
if $\left|v\right|\le n$ for each $\left(q,v,r\right)\in R$. A GJFA
of degree $1$ is called a jumping finite automaton\emph{ }(JFA).
Bold symbols $\mathbf{JFA}$ and $\mathbf{GJFA}$ denote the classes
of languages accepted by these types of automata. 

The present paper contains the following contributions:
\begin{romanlist}
\item We correct erroneous claims from \citep{athMED1} and \citep{athMED1book}
about the closure properties of the class $\mathbf{GJFA}$ -- in fact
it is neither closed under homomorphism and under inverse homomorphism.
\item We answer the open questions about closure properties of $\mathbf{GJFA}$
formulated in these two publications. Specifically, we disprove the
closure under shuffle, Kleene star, and Kleene plus, and prove the
closure under reversal.
\item We relate the new models with the existing ones, pointing out that
the expressive power of GJFA is equivalent to a basic type of graph-controlled
insertion systems, and that the intersection emptiness of GJFA is
undecidable.
\end{romanlist}

\section{Preliminaries}

As described above, a GJFA is a quintuple $M=\left(Q,\Sigma,R,s,F\right)$.
The following formal description of the computation performed by a
GJFA was introduced in \citep{athMED1}.
\begin{definition}
Any string from the language $\Sigma^{*}Q\Sigma^{*}$ is called a
\emph{configuration }of $M$. For $q,r\in Q$ and $t_{1},t_{2},t_{1}',t_{2}',v\in\Sigma^{*}$,
we write 
\[
t_{1}qvt_{2}\curvearrowright_{M}t_{1}'rt_{2}'
\]
if $t_{1}t_{2}=t_{1}'t_{2}'$ and $\left(q,v,r\right)\in R$. By $\curvearrowright_{M}^{*}$
we denote the reflexive-transitive closure of the binary relation
$\curvearrowright_{M}$ over configurations, i.e. $\mathbf{c}\curvearrowright_{M}^{*}\mathbf{c}'$
for configurations $\mathbf{c},\mathbf{c}'$ of $M$ if and only if
$\mathbf{c}=\mathbf{c}'$ or 
\[
\mathbf{d}_{1}\curvearrowright_{M}\cdots\curvearrowright_{M}\mathbf{d}_{i}
\]
for some configurations $\mathbf{d}_{1},\dots,\mathbf{d}_{i}$ of
$M$ with $\mathbf{c}=\mathbf{d}_{1}$, $\mathbf{c}'=\mathbf{d}_{i}$,
$i\ge2$. Finally,
\[
\mathrm{L}\!\left(M\right)=\left\{ u_{1}u_{2}\mid u_{1},u_{2}\in\Sigma^{*},f\in F,u_{1}su_{2}\curvearrowright_{M}^{*}f\right\} 
\]
is the language accepted by $M$. If $M$ is fixed, we write just
$\curvearrowright$ and $\curvearrowright^{*}$.
\end{definition}
The placement of the state symbol $q$ in a configuration $u_{1}qu_{2}$
marks the position of an imaginary tape head. Note that this information
is redundant -- the head is allowed to move anywhere in each step. 

In the present paper we heavily use the natural notion of sequential
insertion, as it was described e.g. in \citep{reaHAU1} and \citep{athKAR2}:
\begin{definition}
Let $K,L\subseteq\Sigma^{*}$ be languages. The \emph{insertion }of
$K$ to $L$ is 
\[
L\leftarrow K=\left\{ u_{1}vu_{2}\mid u_{1}u_{2}\in L,v\in K\right\} .
\]
More generally, for each $k\ge1$ we denote
\begin{eqnarray*}
L\leftarrow^{k}K & = & \left(L\leftarrow^{k-1}K\right)\leftarrow K,\\
L\leftarrow^{*}K & = & \bigcup_{i\ge0}L\leftarrow^{i}K,
\end{eqnarray*}
where $L\leftarrow^{0}K$ stands for $L$. In expressions with $\leftarrow$
and $\leftarrow^{*}$, a singleton set $\left\{ w\right\} $ may be
replaced by $w$. A chain $L_{1}\leftarrow L_{2}\leftarrow\cdots\leftarrow L_{d}$
of insertions is evaluated from the left, e.g. $L_{1}\leftarrow L_{2}\leftarrow L_{3}$
means $\left(L_{1}\leftarrow L_{2}\right)\leftarrow L_{3}$. According
to \citep{reaEHR2}, $L\subseteq\Sigma^{*}$ is a\emph{ unitary language}
if $L=w\leftarrow^{*}K$ for $w\in\Sigma^{*}$ and finite $K\subseteq\Sigma^{*}$.
\end{definition}
Next, we fix additional notation that turns out to be very useful
in our proofs. The notions of \emph{paths }and \emph{labels }naturally
correspond to graphical representations of GJFA, where vertices stand
for states and labeled directed edges stand for rules.
\begin{definition}
\label{def: A-path}Let $M=\left(Q,\Sigma,R,s,F\right)$ be a GJFA.
Each sequence of the form 
\[
\left(q_{0},v_{1},q_{1}\right)\left(q_{1},v_{2},q_{2}\right)\cdots\left(q_{d-1},v_{d},q_{d}\right)\in R^{*}
\]
with $d\ge1$ is a \emph{path }from $q_{0}\in Q$ to $q_{d}\in Q$.
The path is \emph{accepting }if $q_{0}=s$ and $q_{d}\in F$. The
\emph{labeling }of the path is the sequence $v_{1},v_{2},\dots,v_{d}$
of words from $\Sigma^{*}$.
\end{definition}
The lemma below is very natural, its proof only has to deal with formal
combination of different approaches. The symbol $\epsilon$ stands
for the empty word.

\begin{lemma}
\label{conf vs task}Let $M=\left(Q,\Sigma,R,s,F\right)$ be a GJFA.
For each $p,r\in Q$ and $w,w''\in\Sigma^{*}$ the following are equivalent:
\begin{romanlist}
\item \label{enu:conf}$u_{1}pu_{2}\curvearrowright_{M}^{*}u_{1}''ru_{2}''$
for $u_{1},u_{2},u_{1}'',u_{2}''\in\Sigma^{*}$ with $u_{1}u_{2}=w$
and $u_{1}''u_{2}''=w''$.
\item \label{enu:path}$p=r$ and $w=w''$, or 
\begin{equation}
w\in w''\leftarrow v_{d}\leftarrow v_{d-1}\leftarrow\cdots\leftarrow v_{2}\leftarrow v_{1},\label{eq: comp on path}
\end{equation}
where $v_{1},v_{2},\dots,v_{d}$ is a labeling of a path from $p$
to $r$, $d\ge1$.
\end{romanlist}
\end{lemma}
\begin{proof}
First, assume that (\ref{enu:conf}) holds and denote $\mathbf{c}=u_{1}pu_{2}$,
$\mathbf{c}''=u_{1}''ru_{2}''$. If $\mathbf{c}=\mathbf{c}''$, then
$p=r$ and $w=w''$, so we are done. Otherwise, $\mathbf{d}_{1}\curvearrowright_{M}\cdots\curvearrowright_{M}\mathbf{d}_{i}$
for some configurations $\mathbf{d}_{1},\dots,\mathbf{d}_{i}$ of
$M$ with $\mathbf{c}=\mathbf{d}_{1}$, $\mathbf{c}''=\mathbf{d}_{i}$,
$i\ge2$. We use induction by $i$. 

If $i=2$, then $\mathbf{c}\curvearrowright_{M}\mathbf{c}''$ and
according to the definition of $\curvearrowright_{M}$, there are
$t_{1},t_{2},t_{1}'',t_{2}'',v\in\Sigma^{*}$ such that $u_{1}=t_{1}$,
$u_{2}=vt_{2}$, $u_{1}''=t_{1}''$, $u_{2}''=t_{2}''$, and $(p,v,r)\in R$.
Thus, $w\in w''\leftarrow v$ and $\left(p,v,r\right)$ is a path
from $p$ to $r$.

If $i\ge3$, then $\mathbf{c}\curvearrowright_{M}\mathbf{c}'\curvearrowright_{M}^{*}\mathbf{c}''$
for some configuration $\mathbf{c}'=u_{1}'qu_{2}'$. Denote $w'=u_{1}'u_{2}'$.
According to the induction assumption applied to $q,r,w',w''$ we
obtain a path denoted by 
\[
\left(q_{1},v_{2},q_{2}\right)\cdots\left(q_{d-1},v_{d},q_{d}\right)
\]
with $d\ge2$, $q_{1}=q$, $q_{d}=r$, and $w'\in w''\leftarrow v_{d}\leftarrow\cdots\leftarrow v_{2}$.
To conclude the proof it is enough to show that $\left(p,v,q\right)\in R$
and $w\in w'\leftarrow v$ for $v=v_{1}$, which both follow from
$\mathbf{c}\curvearrowright_{M}\mathbf{c}'$ according to the above
analysis of the case $i=2$. 

Second, let (\ref{enu:path}) hold. If $p=r$ and $w=w''$, then $u_{1}pu_{2}=u_{1}''ru_{2}''$
for $u_{1}=u_{1}''=\epsilon$, $u_{2}=w$, and $u_{2}''=w''$. Otherwise,
we fix the path 
\[
\left(q_{0},v_{1},q_{1}\right)\left(q_{1},v_{2},q_{2}\right)\cdots\left(q_{d-1},v_{d},q_{d}\right)
\]
from $q_{0}=p$ to $q_{d}=r$, $d\ge1$, and use induction by $d$.
Denote $v=v_{1}$. 

Let $d=1$. We have $\left(p,v,r\right)\in R$ and $w\in w''\leftarrow v$.
According to the definition of $\leftarrow$, there are $t_{1},t_{2}\in\Sigma^{*}$
such that $w=t_{1}vt_{2}$ and $w''=t_{1}t_{2}$. Obviously, $t_{1}pvt_{2}\curvearrowright_{M}t_{1}rt_{2}$
follows from the definition of $\curvearrowright_{M}$. As $\curvearrowright_{M}$
is a special case of $\curvearrowright_{M}^{*}$, we are done. 

Let $d\ge2$. From (\ref{eq: comp on path}) and the definition of
$\leftarrow$ it follows that 
\begin{equation}
w\in w'\leftarrow v\label{eq: prv}
\end{equation}
for some $w'\in\Sigma^{*}$ with 
\begin{equation}
w'\in w''\leftarrow v_{d}\leftarrow v_{d-1}\leftarrow\cdots\leftarrow v_{2}.\label{eq: zby}
\end{equation}
Denote $q=q_{1}$. According to the induction assumption applied to
$q,r,w',w''$ we obtain 
\[
u_{1}'qu_{2}'\curvearrowright_{M}^{*}u_{1}''ru_{2}''
\]
for some $u_{1}',u_{2}',u_{1}'',u_{2}''\in\Sigma^{*}$ with $u_{1}'u_{2}'=w'$
and $u_{1}''u_{2}''=w''$.

It remains to show that $u_{1}pu_{2}\curvearrowright_{M}u_{1}'qu_{2}'$
for some $u_{1},u_{2}$ with $u_{1}u_{2}=w$. Due to (\ref{eq: prv})
there are $t_{1},t_{2}\in\Sigma^{*}$ such that $w=t_{1}vt_{2}$ and
$w'=t_{1}t_{2}$. From the definition of $\curvearrowright$, together
with $u_{1}'u_{2}'=w'=t_{1}t_{2}$ and $\left(p,v,q\right)\in R$
it follows that 
\[
t_{1}pvt_{2}\curvearrowright_{M}u_{1}'qu_{2}'.
\]
We conclude by denoting $u_{1}=t_{1}$ and $u_{2}=vt_{2}$.
\end{proof}

\begin{corollary}
\label{lem: paths vs accept}Let $M=\left(Q,\Sigma,R,s,F\right)$
be a GJFA and $w\in\Sigma^{*}$. Then $w\in\mathrm{L}\!\left(M\right)$
if and only if $w=\epsilon$ and $s\in F$, or 
\begin{equation}
w\in\epsilon\leftarrow v_{d}\leftarrow v_{d-1}\leftarrow\cdots\leftarrow v_{2}\leftarrow v_{1},\label{eq:comp}
\end{equation}
where $v_{1},v_{2},\dots,v_{d}$ is a labeling of an accepting path
in $M$, $d\ge1$.\end{corollary}
\begin{proof}
If $w\in\mathrm{L}\!\left(M\right)$, then $w=u_{1}u_{2}$ for $u_{1},u_{2}\in\Sigma^{*}$
with $usv\curvearrowright_{M}^{*}f$ and $f\in F$. We apply the forward
implication of Lemma \ref{conf vs task} to $s$, $f$, $w$, and
$\epsilon$. On the other hand, an accepting path ends in some $f\in F$
and we apply the backward implication of Lemma \ref{conf vs task}
to $s$, $f$, $w$, and $\epsilon$.
\end{proof}
The above corollary suggests a generative approach to GJFA -- the
computation of a GJFA may be equivalently described in terms of inserting
factors instead of deleting them. A word is accepted by a GJFA if
and only if it can be composed by inserting factors to the empty word
according to labels of an accepting path. This characterization of
$\mathrm{L}\!\left(M\right)$ will be used very frequently throughout
the paper, so we omit explicit referring to Corollary \ref{lem: paths vs accept}.

Next, we give two simple lemmas that imply the membership in $\mathbf{GJFA}$
for each language that can be described using finite languages and
insertions. 
\begin{lemma}
\label{fin sub GJFA}Each finite language $K\subseteq\Sigma^{*}$
lies in $\mathbf{GJFA}$.\end{lemma}
\begin{proof}
The language $K$ is accepted by the two-state GJFA $M$ with 
\begin{eqnarray*}
M & = & \left\{ \left\{ q,r\right\} ,\Sigma,R,q,\left\{ r\right\} \right\} ,\\
R & = & \left\{ \left(q,w,r\right)\mid w\in K\right\} .
\end{eqnarray*}
Indeed, all the accepting paths in $M$ are of length $1$ and their
labels are exactly the words from $K$. 
\end{proof}

\begin{lemma}
\label{GJFA closed under <- and <-hv}Let $L\subseteq\Sigma^{*}$
lie in $\mathbf{GJFA}$ and $K\subseteq\Sigma^{*}$ be finite. Then
\begin{romanlist}
\item $L\leftarrow K$ lies in $\mathbf{GJFA}$ and
\item $L\leftarrow^{*}K$ lies in $\mathbf{GJFA}$. 
\end{romanlist}
\end{lemma}
\begin{proof}
Let $M_{L}=\left(Q_{L},\Sigma,R_{L},s_{L},F_{L}\right)$ be a GJFA
recognizing $L$. To obtain $M_{1}$ with $\mathrm{L}\!\left(M_{1}\right)=L\leftarrow K$,
we put 
\begin{eqnarray*}
M_{1} & = & \left(Q_{1},\Sigma,R_{1},s,F_{L}\right),\\
Q_{1} & = & Q_{L}\cup\left\{ s\right\} ,\\
R_{1} & = & \left\{ \left(s,v,s_{L}\right)\mid v\in K\right\} \cup R_{L}.
\end{eqnarray*}
First, let $w\in L\leftarrow K$, which means $w\in w'\leftarrow v$
for $w'\in L$ and $v\in K$. As $w'\in L=\mathrm{L}\!\left(M_{L}\right)$,
we have $w'=\epsilon$ and $s_{L}\in F_{L}$, or 
\[
w'\in\epsilon\leftarrow v_{d}\leftarrow\cdots\leftarrow v_{1}
\]
for some accepting path $\mathbf{p}_{L}$ in $M_{L}$ labeled by $v_{1},\dots,v_{d}\in\Sigma^{*}$,
$d\ge1$. In the case of $w'=\epsilon$ and $s_{L}\in F_{L}$, we
have $w=v$ and $w\in\mathrm{L}\!\left(M_{1}\right)$ due to $sv\curvearrowright_{M_{1}}s_{L}$.
Otherwise, observe that $\left(s,v,s_{L}\right)\mathbf{p}_{L}$ is
an accepting path in $M_{1}$ and $w$ lies in
\begin{eqnarray*}
w'\leftarrow v & \subseteq & \left(\epsilon\leftarrow v_{d}\leftarrow\cdots\leftarrow v_{1}\right)\leftarrow v,
\end{eqnarray*}
which matches the labeling of $\left(s,v,s_{L}\right)\mathbf{p}_{L}$.
Thus, $w\in\mathrm{L}\!\left(M_{1}\right)$.

Second, let $w\in\mathrm{L}\!\left(M_{1}\right)$, whence $w=\epsilon$
and $s\in F_{L}$, or 
\[
w\in\epsilon\leftarrow v_{d}\leftarrow v_{d-1}\leftarrow\cdots\leftarrow v_{2}\leftarrow v_{1},
\]
where $v_{1},v_{2},\dots,v_{d}$ is a labeling of an accepting path
in $M_{1}$, $d\ge1$. As $s\notin F_{L}$, we assume the second case.
We have $w\in w'\leftarrow v_{1}$ for some 
\[
w'\in\epsilon\leftarrow v_{d}\leftarrow v_{d-1}\leftarrow\cdots\leftarrow v_{2}.
\]
According to the construction of $M_{1}$, $v_{1}\in K$ and $v_{2},\dots,v_{d}$
is a labeling of an accepting path in $M_{L}$. Thus, $w'\in\mathrm{L}\!\left(M_{L}\right)=L$
and hence $w\in L\leftarrow K$.

To obtain $M_{2}$ with $\mathrm{L}\!\left(M_{2}\right)=L\leftarrow^{*}K$,
we put 
\begin{eqnarray*}
M_{2} & = & \left(Q_{2},\Sigma,R_{2},s,F_{L}\right),\\
Q_{2} & = & Q_{L}\cup\left\{ s\right\} ,\\
R_{2} & = & \left\{ \left(s,v,s\right)\mid v\in K\right\} \cup\left\{ \left(s,\epsilon,s_{L}\right)\right\} \cup R_{L}.
\end{eqnarray*}
First, let $w\in L\leftarrow^{*}K$, which means $w\in w'\leftarrow^{i}K$
for $w'\in L$, $i\ge0$.
\begin{romanlist}
\item If $i=0$, then $w=w'\in L$ and we show that $L\subseteq\mathrm{L}\!\left(M_{2}\right)$.
Indeed, due to Corollary \ref{lem: paths vs accept} for each $w'\in L$
it holds that $w'=\epsilon$ and $s_{L}\in F_{L}$, or 
\begin{equation}
w'\in\epsilon\leftarrow v_{d}\leftarrow v_{d-1}\leftarrow\cdots\leftarrow v_{2}\leftarrow v_{1},\label{eq:comp-1}
\end{equation}
for some accepting path $\mathbf{p}_{L}$ in $M_{L}$ labeled by $v_{1},\dots,v_{d}\in\Sigma^{*}$,
$d\ge1$. If $w'=\epsilon$ and $s_{L}\in F_{L}$, then $w'\in\mathrm{L}\!\left(M_{2}\right)$
due to Corollary \ref{lem: paths vs accept} applied to the path $\left(s,\epsilon,s_{L}\right)$
together with $w'\in\epsilon\leftarrow\epsilon$. Otherwise, $w'\in\mathrm{L}\!\left(M_{2}\right)$
due to Corollary \ref{lem: paths vs accept} applied to the path $\left(s,\epsilon,s_{L}\right)\mathbf{p}_{L}$
together with
\begin{eqnarray*}
w' & \in & \epsilon\leftarrow v_{d}\leftarrow v_{d-1}\leftarrow\cdots\leftarrow v_{2}\leftarrow v_{1},\\
 & = & \epsilon\leftarrow v_{d}\leftarrow v_{d-1}\leftarrow\cdots\leftarrow v_{2}\leftarrow v_{1}\leftarrow\epsilon.
\end{eqnarray*}

\item Let $i\ge1$. According to the definition of $\leftarrow^{i}$, we
have 
\begin{eqnarray*}
w & \in & w'\leftarrow v_{i}\leftarrow\cdots\leftarrow v_{1},
\end{eqnarray*}
where $v_{1},\dots,v_{i}\in K$ and $w'\in L$. As $L\subseteq\mathrm{L}\!\left(M_{2}\right)$
and $s\notin F_{L}$, there is an accepting path $\mathbf{p}_{L}$
in $M_{2}$ labeled by $v_{1}',\dots,v_{d}'$ with
\[
w'\in\epsilon\leftarrow v_{d}'\leftarrow\cdots\leftarrow v_{1}'.
\]
It remains to observe that 
\[
\left(s,v_{1},s\right)\cdots\left(s,v_{i},s\right)\left(s,\epsilon,s_{L}\right)\mathbf{p}_{L}
\]
is an accepting path in $M_{2}$ and 
\begin{eqnarray*}
w & \in & w'\leftarrow v_{i}\leftarrow\cdots\leftarrow v_{1}\\
 & \subseteq & \epsilon\leftarrow v_{d}'\leftarrow\cdots\leftarrow v_{1}'\leftarrow\epsilon\leftarrow v_{i}\leftarrow\cdots\leftarrow v_{1}.
\end{eqnarray*}

\end{romanlist}

\noindent Second, let $w\in\mathrm{L}\!\left(M_{2}\right)$, whence
$w=\epsilon$ and $s\in F_{L}$, or 
\[
w\in\epsilon\leftarrow v_{d}\leftarrow v_{d-1}\leftarrow\cdots\leftarrow v_{2}\leftarrow v_{1},
\]
such that there is an accepting path 
\[
\mathbf{p}=\left(q_{0},v_{1},q_{1}\right)\cdots\left(q_{d-1},v_{d},q_{d}\right)
\]
in $M_{2}$, $d\ge1$. As $s\notin F_{L}$, we assume the second case.
According to the construction of $M_{2}$, there is a unique $1\le c\le d$
such that $\left(q_{c-1},v_{c},q_{d}\right)=\left(s,\epsilon,s_{L}\right)$.
Denote $\mathbf{p}=\mathbf{p}_{K}\left(s,\epsilon,s_{L}\right)\mathbf{p}_{L}$,
where $\mathbf{p}_{K}$ is labeled by words from $K$ and $\mathbf{p}_{L}$
is an accepting path in $M_{L}$. Both $\mathbf{p}_{K}$ and $\mathbf{p}_{L}$
may stand for empty strings, their lengths are $c-1$ and $d-c$ respectively.
We get
\[
w\in w'\leftarrow\epsilon\leftarrow v_{c-1}\leftarrow\cdots\leftarrow v_{2}\leftarrow v_{1},
\]
where
\[
w'\in\epsilon\leftarrow v_{d}\leftarrow\cdots\leftarrow v_{c+1}\subseteq L.
\]
Together, $w\in L\leftarrow^{c-1}K\subseteq L\leftarrow^{*}K$.

\end{proof}
Let us give a few examples of GJFA languages that are used later in
this paper and follow easily from the above lemmas. Note that a GJFA
over an alphabet $\Sigma$ can be seen as operating over any alphabet
$\Sigma'\supseteq\Sigma$. 
\begin{example}
\label{The-following-languages}The following languages lie in $\mathbf{GJFA}$:
\begin{romanlist}
\item The trivial language $\Sigma^{*}=\epsilon\leftarrow^{*}\Sigma$ over
an arbitrary $\Sigma$.
\item The language $\Sigma^{*}u\Sigma^{*}=\Sigma^{*}\leftarrow u$ for $u\in\Sigma^{*}$
over an arbitrary $\Sigma$.
\item The Dyck language $D$ over $\Sigma=\left\{ a,\overline{a}\right\} $.
We have $D=\epsilon\leftarrow^{*}a\overline{a}$.
\item Any semi-Dyck language $D_{k}$ over $\Sigma=\left\{ a_{1},\dots,a_{k},\overline{a}_{1},\dots,\overline{a}_{k}\right\} $.
We have $D_{k}\leftarrow^{*}\left\{ a_{1}\overline{a}_{1},\dots,a_{k}\overline{a}_{k}\right\} $. 
\item Any unitary language.
\end{romanlist}
\end{example}
However, there are GJFA languages that cannot be simply obtained from
finite languages by applying Lemma \ref{GJFA closed under <- and <-hv},
such as the following classical language that is not context-free
and lies even in $\mathbf{JFA}$. By $\left|w\right|_{x}$ we denote
the number of occurrences of a letter $x\in\Sigma$ in a word $w\in\Sigma^{*}$.
\begin{example}
The JFA $M$ with
\begin{eqnarray*}
M & = & \left(\left\{ q_{0},q_{1},q_{2}\right\} ,\Sigma,R,q_{0},\left\{ q_{0}\right\} \right),\\
R & = & \left\{ \left(q_{0},a,q_{1}\right),\left(q_{1},b,q_{2}\right),\left(q_{2},c,q_{0}\right)\right\} 
\end{eqnarray*}
accepts the language $L=\left\{ w\in\Sigma^{*}|\left|w\right|_{a}=\left|w\right|_{b}=\left|w\right|_{c}\right\} $
over $\Sigma=\left\{ a,b,c\right\} ^{*}$.
\end{example}
The above example shows that the class $\mathbf{GJFA}$ is not a subclass
of context-free languages, but it was pointed out in \citep{athMED1}
that each GJFA language is context-sensitive. The class $\mathbf{GJFA}$
does not stick to classical measures of expressive power -- in the
next section we give examples of regular languages that do not lie
in $\mathbf{GJFA}$. As for JFA languages, in \citep{athMED1book}
the authors show that a language lies in $\mathbf{JFA}$ if and only
if it is equal to the permutation closure of a regular language.

\section{A Necessary Condition for Membership in GJFA}

In order to formulate our main tools for disproving membership in
GJFA, the following technical notions remain to be defined.
\begin{definition}
A language $K\subseteq\Sigma^{*}$ is a \emph{composition }if $K=\left\{ \epsilon\right\} $
or 
\[
K=\epsilon\leftarrow v_{d}\leftarrow v_{d-1}\leftarrow\cdots\leftarrow v_{2}\leftarrow v_{1},
\]
for some $v_{1},\dots,v_{d}\in\Sigma^{*}$, $d\ge1$. A composition
$K$ is of \emph{degree} $n\ge0$ if $K=\left\{ \epsilon\right\} $
or$\left|v_{i}\right|\le n$ for each $i\in\left\{ 1,\dots,d\right\} $.
For each $n\ge0$, let $\mathbf{UC}_{n}$ denote the class of languages
$L$ that can be written as 
\[
L=\bigcup_{K\in\mathcal{C}}K,
\]
where $\mathcal{C}$ is any (possibly infinite) set of compositions
of degree $n$. We also denote $\mathbf{UC}=\bigcup_{n\ge0}\mathbf{UC}_{n}$.
\end{definition}

The class\textbf{ $\mathbf{UC}$ }itself does not seem to be of practical
importance -- we use the membership in $\mathbf{UC}$ only as a necessary
condition for membership in $\mathbf{GJFA}$. The acronym \textbf{UC}
stands for \emph{union of compositions}\@.
\begin{lemma}
\label{Lem:neces zob}$\mathbf{GJFA}\subseteq\mathbf{UC}$.\end{lemma}
\begin{proof}
Let $M=\left(Q,\Sigma,R,s,F\right)$ be a GJFA. Let $\mathcal{P}$
be the set of all accepting paths in $M$. According to Corollary
\ref{lem: paths vs accept}, we have 
\[
\mathrm{L}\!\left(M\right)\backslash\left\{ \epsilon\right\} =\bigcup_{\mathbf{p}\in\mathcal{P}}\left(\epsilon\leftarrow v_{\mathbf{p},d_{\mathbf{p}}}\leftarrow v_{\mathbf{p},d_{\mathbf{p}}-1}\leftarrow\cdots\leftarrow v_{\mathbf{p},2}\leftarrow v_{\mathbf{p},1}\right),
\]
where $v_{\mathbf{p},1},\dots,v_{\mathbf{p},d_{\mathbf{p}}}$ is the
labeling of $\mathbf{p}$, $d_{\mathbf{p}}\ge1$. As $\left\{ \epsilon\right\} $
is a composition, we have $\mathrm{L}\!\left(M\right)\in\mathbf{UC}_{n}$,
where $n=\max\left\{ \left|v\right|\mid q,r\in Q,\left(q,v,r\right)\in R\right\} $.
\end{proof}
The following lemma deals with the language $L=\left\{ ab\right\} ^{*}$,
which serves as a canonical non-GJFA language in the proofs of our
main results.
\begin{lemma}
\label{lem:gjfa ab star}The language $L=\left\{ ab\right\} ^{*}$
does not lie in $\mathrm{\mathbf{GJFA}}$.\end{lemma}
\begin{proof}
Assume for a contradiction that $L\in\mathbf{GJFA}$. Due to Lemma
\ref{Lem:neces zob}, $L\in\mathbf{UC}$ and thus $L\in\mathbf{UC}_{n}$
for some $n\ge0$. If $n=0$, observe that $L=\left\{ \epsilon\right\} $,
which is a contradiction. Otherwise, fix $w=\left(ab\right)^{n+1}$.
According to the definition of $\mathbf{UC}_{n}$ , $w$ lies in a
composition $K\subseteq L$ of the form
\[
K=\epsilon\leftarrow v_{d}\leftarrow v_{d-1}\leftarrow\cdots\leftarrow v_{2}\leftarrow v_{1}
\]
of degree $n$. Due to $w\neq\epsilon$, there exists the least $c$
with $d\ge c\ge1$ and $v_{c}\neq\epsilon$. Moreover,
\begin{eqnarray*}
K & = & K'\leftarrow v
\end{eqnarray*}
for suitable $K'$ and $v=v_{c}$. Thus, $w=u_{1}vu_{2}$ for $u_{1}u_{2}\in K'$.
As $\left|v\right|\le n$, at least one of the following cases holds:
\begin{romanlist}
\item Assume that $\left|u_{1}\right|\ge2$ and write $u_{1}=ab\overline{u}_{1}$.
If $v$ starts by $a$, we have $avb\overline{u}_{1}u_{2}\in K$.
If $v$ starts by $b$, we have $abv\overline{u}_{1}u_{2}\in K$.
\item Assume that $\left|u_{2}\right|\ge2$ and write $u_{2}=\overline{u}_{2}ab$.
If $v$ starts by $a$, we have $u_{1}\overline{u}_{2}avb\in K$.
If $v$ starts by $b$, we have $u_{1}\overline{u}_{2}abv\in K$.
\end{romanlist}

In each case, $K$ contains a word having some of the factors $aa$,
$bb$. Thus $K\nsubseteq L$, which is a contradiction. Informally,
each $w$ coming from a language $L\in\mathbf{UC}_{n}$ must contain
a factor $v$ of length at most $n$ that can be inserted to any other
place in $u_{1}u_{2}$ such that the result stays in $L$. In the
case of $L=\left\{ ab\right\} ^{*}$ this property fails.

\end{proof}

\section{Closure Properties of GJFA Languages}

The table below lists various unary and binary operators on languages.
The symbols $+,-$ tell that a class is closed or is not closed under
an operator, respectively. A similar table was presented in \citep{athMED1,athMED1book},
containing several question marks. In this section we complete and
correct these results. The symbol $\diamondsuit$ marks answers to
open questions and the symbol $\blacklozenge$ marks corrections.

\begin{center}
\begin{tabular}{|l|>{\raggedleft}p{10mm}>{\centering}p{3mm}|>{\centering}p{17mm}|}
\cline{2-4} 
\multicolumn{1}{l|}{} & \multicolumn{2}{c|}{$\mathbf{GJFA}$} & $\mathbf{JFA}$\tabularnewline
\hline 
Endmarking & $-$ &  & $-$\tabularnewline
\hline 
Concatenation & $-$ &  & $-$\tabularnewline
\hline 
Shuffle & $-$ & $\diamondsuit$ & $+$\tabularnewline
\hline 
Union & $+$ &  & $+$\tabularnewline
\hline 
Complement & $-$ &  & $+$\tabularnewline
\hline 
Intersection & $-$ &  & $+$\tabularnewline
\hline 
Int. with regular languages & $-$ &  & $-$\tabularnewline
\hline 
Kleene star & $-$ & $\diamondsuit$ & $-$\tabularnewline
\hline 
Kleene plus & $-$ & $\diamondsuit$ & $-$\tabularnewline
\hline 
Reversal & $+$ & $\diamondsuit$ & $+$\tabularnewline
\hline 
Substitution & $-$ &  & $-$\tabularnewline
\hline 
Regular substitution & $-$ &  & $-$\tabularnewline
\hline 
Finite substitution & $-$ & $\blacklozenge$ & $-$\tabularnewline
\hline 
Homomorphism & $-$ & $\blacklozenge$ & $-$\tabularnewline
\hline 
$\epsilon$-free homomorphism & $-$ & $\blacklozenge$ & $-$\tabularnewline
\hline 
Inverse homomorphism & $-$ & $\blacklozenge$ & $+$\tabularnewline
\hline 
\end{tabular}\medskip{}

\par\end{center}

Before proving the new results, let us deal with the closure under
intersection. The authors of \citep{athMED1,athMED1book} claim that
the theorem below follows from an immediate application of De Morgan's
laws to the results about union and complement. We find this argument
invalid and present an explicit proof of the claim.
\begin{theorem}
\label{inters}$\mathbf{GJFA}$ is not closed under intersection.\end{theorem}
\begin{proof}
Let $\Sigma=\left\{ a,\overline{a}\right\} $ and $L=\mathrm{L}\!\left(M\right)$
for 
\begin{eqnarray*}
M & = & \left(\left\{ q,r\right\} ,\Sigma,R,q,\left\{ r\right\} \right),\\
R & = & \left\{ \left(q,\overline{a}a,q\right),\left(q,a\overline{a},r\right)\right\} ,
\end{eqnarray*}
as depicted in Figure \ref{fig:The-GJFA M inters}. For each $d\ge1$
there is exactly one accepting path of length $d$ in $M$. According
to Corollary \ref{lem: paths vs accept}, we have 
\[
L=\bigcup_{d\ge1}K_{d},
\]
where $K_{1}=\epsilon\leftarrow a\overline{a}$ and $K_{i+1}=K_{i}\leftarrow\overline{a}a$
for $i\ge1$. We show that 
\begin{equation}
D\cap L=\left\{ a\overline{a}\right\} ^{*},\label{eq: intersection}
\end{equation}
where $D\in\mathbf{GJFA}$ is the Dyck language from Example \ref{The-following-languages},
and $\left\{ a\overline{a}\right\} ^{*}$ does not lie in\textbf{
$\mathbf{GJFA}$} due to Lemma \ref{lem:gjfa ab star}. The backward
inclusion is easy. As for the forward one, we have 
\begin{eqnarray*}
D\cap L & = & D\cap\bigcup_{d\ge1}K_{d}\\
 & = & \bigcup_{d\ge1}\left(D\cap K_{d}\right),
\end{eqnarray*}
so it is enough to verify that $D\cap K_{d}\subseteq\left\{ a\overline{a}\right\} ^{*}$
for each $d\ge1$. The case $d=1$ is trivial since $K_{1}=\left\{ a\overline{a}\right\} $.
In order to continue inductively, fix $d\ge2$. For any $w\in D\cap K_{d}$,
we have $w=u_{1}\overline{a}au_{2}$ for $u_{1}u_{2}\in K_{d-1}$.
From $D=\epsilon\leftarrow^{*}a\overline{a}$ it follows that $u_{1}\in DaD$,
$u_{2}\in D\overline{a}D$, and thus, $u_{1}u_{2}\in D$. By the induction
assumption, $u_{1}u_{2}\in\left\{ a\overline{a}\right\} ^{*}$. Hence
$w\in\left\{ a\overline{a}\right\} ^{*}\left(\overline{a}a\right)\left\{ a\overline{a}\right\} ^{*}$
or $w\in\left\{ a\overline{a}\right\} ^{*}a\left(\overline{a}a\right)\overline{a}\left\{ a\overline{a}\right\} ^{*}$.
The first case implies $w\notin D$, which is a contradiction, and
the second case implies $w\in\left\{ a\overline{a}\right\} ^{*}$. 
\end{proof}
The next theorem shows that some of our results actually follow very
easily from Lemma \ref{lem:gjfa ab star}, which claims that $\left\{ ab\right\} ^{*}\notin\mathbf{GJFA}$.
Theorems \ref{inv hom} and \ref{thm: GJFA unary shuffle} provide
special counter-examples for the closure under inverse homomorphism
and shuffle.
\begin{figure}
\begin{centering}
\includegraphics{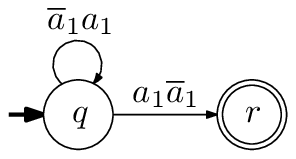}
\par\end{centering}

\caption{\label{fig:The-GJFA M inters}The GJFA $M$ with $D\cap\mathrm{L}\!\left(M\right)=\left\{ a\overline{a}\right\} ^{*}$}
\end{figure}
\begin{figure}
\begin{centering}
\includegraphics{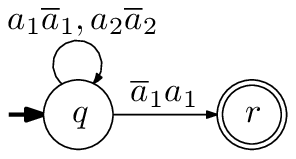}
\par\end{centering}

\caption{\label{fig:The-GJFA M}The GJFA $M$ with $\mathrm{L}\!\left(M\right)=D_{2}\overline{a}_{1}D_{2}a_{1}D_{2}$}
\end{figure}

\begin{theorem}
$\mathbf{GJFA}$ is not closed under:\end{theorem}
\begin{romanlist}
\item Kleene star,
\item Kleene plus,
\item $\epsilon$-free homomorphism,
\item homomorphism,
\item finite substitution.\end{romanlist}
\begin{proof}
We have $\left\{ ab\right\} \in\mathbf{GJFA}$ and $\left\{ ab\right\} ^{*}\notin\mathbf{GJFA}$
due to Lemma \ref{lem:gjfa ab star}. As $\mathbf{GJFA}$ is closed
under union, $\left\{ ab\right\} ^{+}\notin\mathbf{GJFA}$ as well.
As for $\epsilon$-free homomorphism, consider $\varphi:\left\{ a\right\} ^{*}\rightarrow\left\{ a,b\right\} ^{*}$
with $\varphi\!\left(a\right)=ab$. We have $L=\left\{ a\right\} ^{*}\in\mathbf{GJFA}$
and $\varphi\!\left(L\right)=\left\{ ab\right\} ^{*}\notin\mathbf{GJFA}$.
Trivially, $\varphi$ is also a general homomorphism and a finite
substitution.\end{proof}
\begin{theorem}
\label{inv hom}$\mathbf{GJFA}$ is not closed under inverse homomorphism. \end{theorem}
\begin{proof}
Let $\Sigma=\left\{ a_{1},\overline{a}_{1},a_{2},\overline{a}_{2}\right\} $
and 
\begin{eqnarray*}
M & = & \left(\left\{ q,r\right\} ,\Sigma,R,q,\left\{ r\right\} \right),\\
R & = & \left\{ \left(q,a_{1}\overline{a}_{1},q\right),\left(q,a_{2}\overline{a}_{2},q\right),\left(q,\overline{a}_{1}a_{1}r\right)\right\} ,
\end{eqnarray*}
see Figure \ref{fig:The-GJFA M}. Let $L=\mathrm{L}\!\left(M\right)$.
Observe that $L=D_{2}\overline{a}_{1}D_{2}a_{1}D_{2},$ where $D_{2}$
is the semi-Dyck language with two types of brackets: $a_{1},\overline{a}_{1}$
and $a_{2},\overline{a}_{2}$. According to Example~\ref{The-following-languages},
$D_{2}\in\mathbf{GJFA}$. Let $\varphi:\left\{ a,b\right\} ^{*}\rightarrow\Sigma^{*}$
be defined as
\begin{eqnarray*}
\varphi\!\left(a\right) & = & \overline{a}_{1}a_{2},\\
\varphi\!\left(b\right) & = & \overline{a}_{2}a_{1},
\end{eqnarray*}
and let us claim that 
\[
\varphi^{-1}\!\left(L\right)=\left\{ ab\right\} ^{*},
\]
which is not a GJFA language according to Lemma \ref{lem:gjfa ab star}.

The backward inclusion is easy -- for each $v=\left(ab\right)^{i}$
with $i\ge0$ we have $\varphi\!\left(v\right)=\overline{a}_{1}\left(a_{2}\overline{a}_{2}\right)^{i-1}a_{1}$
if $i\ge1$, and $\varphi\!\left(v\right)=\epsilon$ if $i=0$. In
both cases, $\varphi\!\left(v\right)\in L$ and thus $v\in\varphi^{-1}\!\left(L\right)$. 

As for the forward inclusion, take any $v\in\varphi^{-1}\!\left(L\right)$
and fix $w\in L$ with $\varphi\!\left(v\right)=w$. As $w\in L$,
we have $w=u_{1}u_{2}u_{3}$ for $u_{1},u_{3}\in D_{2}$ and $u_{2}\in\overline{a}_{1}D_{2}a_{1}$.
As $w\in\mathrm{range}\!\left(\varphi\right)$, $w$ starts with $\overline{a}_{1}$
or $\overline{a}_{2}$ and ends with $a_{1}$ or $a_{2}$. Because
$u_{1}\in D_{2}$ cannot start with $\overline{a}_{1}$ nor $\overline{a}_{2}$
and end $u_{3}\in D_{2}$ cannot end with $a_{1}$ nor $a_{2}$, we
have $u_{1}=u_{3}=\epsilon$ and $w\in\overline{a}_{1}D_{2}a_{1}$.
Denote $v=x_{1}\cdots x_{m}$ for $x_{1},\dots,x_{m}\in\left\{ a,b\right\} $.
As $w\in\overline{a}_{1}D_{2}a_{1}$, $x_{1}=a$, $x_{m}=b$, and
\[
w=\overline{a}_{1}a_{2}\varphi\!\left(x_{2}\right)\cdots\varphi\!\left(x_{m-1}\right)\overline{a}_{2}a_{1},
\]
where
\begin{eqnarray*}
a_{2}\varphi\!\left(x_{2}\right)\cdots\varphi\!\left(x_{m-1}\right)\overline{a}_{2} & \in & D_{2}.
\end{eqnarray*}
None of the factors $a_{2}\overline{a}_{1}$ and $a_{1}\overline{a}_{2}$
can occur in $D_{2}$. It follows that $x_{2}=b$ and for each $i=2,\dots,m-2$
it holds that 
\[
x_{i}=a\,\Leftrightarrow\, x_{i+1}=b,
\]
which together implies $v\in\left\{ ab\right\} ^{*}$.
\end{proof}
For languages $K,L\subseteq\Sigma^{*}$, a word $w\in\Sigma^{*}$
belongs to $\mathrm{shuffle}\!\left(K,L\right)$ if and only if there
are words $u_{1},u_{2},\dots,u_{k},v_{1},v_{2},\dots,v_{k}\in\Sigma^{*}$
such that $u_{1}u_{2}\cdots u_{k}\in K$, $v_{1}v_{2}\cdots v_{k}\in L$,
and $w=u_{1}v_{1}u_{2}v_{2}\cdots u_{k}v_{k}$. Furthermore, we denote
$L^{\mathrm{R}}=\left\{ w^{\mathrm{R}}\mid w\in L\right\} $, where
$w^{\mathrm{R}}$ is the reversal of $w$.
\begin{theorem}
\label{thm: GJFA unary shuffle}$\mathbf{GJFA}$ is not closed under
shuffle.\end{theorem}
\begin{proof}
Again, we fix $\Sigma=\left\{ a_{1},\overline{a}_{1},a_{2},\overline{a_{2}}\right\} $
and consider the semi-Dyck language $L=D_{2}\in\mathbf{GJFA}$ over
$\Sigma$. We claim that $\mathrm{shuffle}\!\left(D_{2},D_{2}\right)\notin\mathbf{GJFA}$.
According to Lemma \ref{Lem:neces zob} we assume for a contradiction
that $\mathrm{shuffle}\!\left(D_{2},D_{2}\right)\in\mathbf{UC}_{n}$
for $n\ge1$. Denote $w=a_{1}^{n}a_{2}^{n}\overline{a}_{1}^{n}\overline{a}_{2}^{n}$.
The word $w$ lies in a composition $K$ of degree $n$ having the
form $K=K'\leftarrow v$, so $w=u_{1}vu_{2}$ for $u_{1}u_{2}\in K'$.
Clearly, there is $x\in\left\{ a_{1},a_{2}\right\} $ such that at
least one of the following assumptions is fulfilled:
\begin{romanlist}
\item Assume that $v$ contains $x$. As $\left|v\right|\le n$, it cannot
contain $\overline{x}$. The word $u_{1}u_{2}v$ lies in $K$ but
it contains an occurrence of $x$ with no occurrence of $\overline{x}$
on the right, so it does not lie in $\mathrm{shuffle}\!\left(D_{2},D_{2}\right)$.
\item Assume that $v$ contains $\overline{x}$. As $\left|v\right|\le n$,
it cannot contain $x$. The word $vu_{1}u_{2}$ lies in $K$ but it
contains an occurrence of $\overline{x}$ with no occurrence $x$
on the left, so it does not lie in $\mathrm{shuffle}\!\left(D_{2},D_{2}\right)$.
\end{romanlist}
\end{proof}
\begin{lemma}
\label{rev ins}For each $K,L\subseteq\Sigma^{*}$ it holds that $\left(L\leftarrow K\right)^{\mathrm{R}}=L^{\mathrm{R}}\leftarrow K^{\mathrm{R}}$.\end{lemma}
\begin{proof}
First, let $w\in\left(L\leftarrow K\right)^{\mathrm{R}}$, which means
$w=\left(u_{1}vu_{2}\right)^{\mathrm{R}}$ for $v\in K$ and $u_{1}u_{2}\in L$.
We have $\left(u_{1}vu_{2}\right)^{\mathrm{R}}=u_{2}^{\mathrm{R}}v^{\mathrm{R}}u_{1}^{\mathrm{R}}$,
while $v^{\mathrm{R}}\in K^{\mathrm{R}}$ and $u_{2}^{\mathrm{R}}u_{1}^{\mathrm{R}}=\left(u_{1}u_{2}\right)^{\mathrm{R}}\in L^{\mathrm{R}}$.
Thus, $w\in L^{\mathrm{R}}\leftarrow K^{\mathrm{R}}$. Second, let
$w\in L^{\mathrm{R}}\leftarrow K^{\mathrm{R}}$, which means $w=u_{1}vu_{2}$
for $v\in K^{\mathrm{R}}$ and $u_{1}u_{2}\in L^{\mathrm{R}}$. Thus,
$v^{\mathrm{R}}\in K$ and $\left(u_{1}u_{2}\right)^{\mathrm{R}}=u_{2}^{\mathrm{R}}u_{1}^{\mathrm{R}}\in L$.
As $w^{\mathrm{R}}=u_{2}^{\mathrm{R}}v^{\mathrm{R}}u_{1}^{\mathrm{R}}$,
it follows that $w^{\mathrm{R}}\in L\leftarrow K$ and thus $w\in\left(L\leftarrow K\right)^{\mathrm{R}}$. \end{proof}
\begin{theorem}
\label{thm:gjfa reversal}$\mathbf{GJFA}$ is closed under reversal.\end{theorem}
\begin{proof}
For arbitrary GJFA $M=\left(Q,\Sigma,R,s,F\right)$, we define a GJFA
$\mathrm{rev}\!\left(M\right)$ as follows: 
\begin{eqnarray*}
\mathrm{rev}\!\left(M\right) & = & \left(Q,\Sigma,R',s,F\right),\\
R' & = & \left\{ \left(q,v^{\mathrm{R}},r\right)\mid\left(q,v,r\right)\in R\right\} .
\end{eqnarray*}
Trivially, $\mathrm{rev}\!\left(\mathrm{rev}\!\left(M\right)\right)=M$
for each GJFA $M$. We claim that 
\[
\mathrm{L}\!\left(\mathrm{rev}\!\left(M\right)\right)=\mathrm{L}\!\left(M\right)^{\mathrm{R}},
\]
which means that $\mathrm{L}\!\left(M\right)^{\mathrm{R}}$ is always
a GJFA language.

First, we show $\mathrm{L}\!\left(M\right)^{\mathrm{R}}\subseteq\mathrm{L}\!\left(\mathrm{rev}\!\left(M\right)\right)$.
Let $w\in\mathrm{L}\!\left(M\right)^{\mathrm{R}}$, i.e. $w^{\mathrm{R}}\in\mathrm{L}\!\left(M\right)$,
which means that $w^{\mathrm{R}}=\epsilon$ and $s\in F$, or 
\begin{equation}
w^{\mathrm{R}}\in\epsilon\leftarrow v_{d}\leftarrow v_{d-1}\leftarrow\cdots\leftarrow v_{2}\leftarrow v_{1},\label{eq:comp-2}
\end{equation}
where $v_{1},v_{2},\dots,v_{d}$ is a labeling of an accepting path
in $M$, $d\ge1$. If $w^{\mathrm{R}}=\epsilon$ and $s\in F$, it
follows easily that $w\in\mathrm{L}\!\left(\mathrm{rev}\!\left(M\right)\right)$.
Otherwise, observe that 
\begin{eqnarray*}
w & \in & \left(\epsilon\leftarrow v_{d}\leftarrow v_{d-1}\leftarrow\cdots\leftarrow v_{2}\leftarrow v_{1}\right)^{\mathrm{R}}\\
 & = & \epsilon\leftarrow v_{d}^{\mathrm{R}}\leftarrow v_{d-1}^{\mathrm{R}}\leftarrow\cdots\leftarrow v_{2}^{\mathrm{R}}\leftarrow v_{1}^{\mathrm{R}}
\end{eqnarray*}
according to Lemma \ref{rev ins}, and $v_{1}^{\mathrm{R}},v_{2}^{\mathrm{R}},\dots,v_{d}^{\mathrm{R}}$
is a labeling of an accepting path in $\mathrm{rev}\!\left(M\right)$.
Together, $w^{\mathrm{R}}\in\mathrm{L}\!\left(\mathrm{rev}\!\left(M\right)\right)$. 

Second, we show $\mathrm{L}\!\left(\mathrm{rev}\!\left(M\right)\right)\subseteq\mathrm{L}\!\left(M\right)^{\mathrm{R}}$.
According to the first inclusion applied to $\mathrm{rev}\!\left(M\right)$
instead of $M$, it holds that

\begin{eqnarray*}
\mathrm{L}\!\left(\mathrm{rev}\!\left(M\right)\right)^{\mathrm{R}} & \subseteq & \mathrm{L}\!\left(\mathrm{rev}\!\left(\mathrm{rev}\!\left(M\right)\right)\right)\\
 & = & \mathrm{L}\!\left(M\right),
\end{eqnarray*}
which is trivially equivalent to $\mathrm{L}\!\left(\mathrm{rev}\!\left(M\right)\right)\subseteq\mathrm{L}\!\left(M\right)^{\mathrm{R}}$.
\end{proof}

\section{\label{sec:Relations-to-Other}Relations to the Other Models}

An \emph{insertion-deletion system}, as introduced in \citep{reaPAU4},
generates words from a finite set of axioms by nondeterministic inserting
and deleting factors according to insertion rules and deletion rules.
An insertion or deletion rule may specify left and right contexts
that are needed to perform the operation. Such a system is an \emph{insertion
system }if there are no deletion rules. The language accepted by an
insertion system contains all the words that can be obtained from
the axioms using the insertion rules. Generally, an insertion-deletion
system may use nonterminals, but this does not make sense for insertion
systems.

For each $n,m,m'\ge0$, the term $\mathrm{ins}_{n}^{m,m'}$ denotes
the class of languages accepted by insertion systems where each insertion
rule inserts a factor of length at most $n$ and depends on left and
right contexts of lengths at most $m,m'$ respectively. The symbol
$\ast$ says that a parameter is not bounded. For example, $\mathrm{ins}_{\ast}^{0,0}$
contains exactly finite unions of unitary languages (one unitary language
for each axiom).

In \citep{reaALH1} and \citep{reaVER1}, the authors introduce \emph{graph-controlled
insertion systems}. Such systems may be described by a set of axioms
and a directed multigraph with edges labeled by insertion rules. The
vertices are called \emph{components}. An initial component and a
final component are specified. A graph-controlled insertion system
accepts a word if and only if the word is obtained from an axiom using
a sequence of insertion rules that forms a path from the initial component
to the final component.

For each $n,m,m',k\ge0$, the term $\mathrm{LStP}_{k}\!\left(\mathrm{ins}_{n}^{m,m'}\right)$
denotes the class of languages accepted by graph-regulated insertion
systems with at most $k$ components where the properties of insertion
rules are bounded by $n,m,m'$ as above. It turns out that\emph{ 
\[
\mathbf{GJFA}=\mathrm{LStP}_{\ast}\!\left(\mathrm{ins}_{\ast}^{0,0}\right).
\]
}Indeed, with the generative approach to GJFA in mind, a GJFA may
be transformed to a graph-regulated insertion system with the same
structure, using only the axiom $\epsilon$ and insertion rules with
empty contexts. For the backward inclusion we just encode the axioms
to rules specifying that the computation ends by deleting an axiom.

Other related models were introduced in 1980's in the scope of Galiukschov
semicontextual grammars, see \citep{reaPAU6}. For $k\ge0$, a \emph{semicontextual
grammar of degree $k$ without appearance checking} is actually an
insertion system with left and right contexts of length at most $k$.
Moreover, such grammar may involve a \emph{regular control}. In this
case, a regular language $C$ over the alphabet of insertion rules
is specified, and a string is accepted only if it can be obtained
from an axiom using a sequence of insertion rules that belongs to
$C$. The symbol $\mathcal{C}_{k}$ denotes the class of languages
accepted by \emph{regular control semicontextual grammars of degree
$k$ without appearance checking}. The equivalence of $\mathcal{C}_{0}$
and $\mathrm{LStP}_{\ast}\!\left(\mathrm{ins}_{\ast}^{0,0}\right)$
is immediate. Thus 
\[
\mathbf{GJFA}=\mathcal{C}_{0}.
\]

Finally, an interesting result follows from the fact that each unitary
language lies in $\mathbf{GJFA}$. According the main result of the
Haussler's article \citep{reaHAU1}, for a given a prefix-disjoint
instance of the Post correspondence problem (PCP)%
\footnote{We do not introduce PCP nor prefix-disjoint instances in this paper.%
} over a range alphabet $\Sigma$, one can algorithmically construct
sets $S,T\subseteq\Gamma^{*}$ over an alphabet $\Gamma$ with $\left|\Gamma\right|=2\left|\Sigma\right|+4$
such that the PCP instance is positive if and only if the intersection
of $\epsilon\leftarrow^{*}S$ and $\epsilon\leftarrow^{*}T$ contains
a non-empty string. For given $S$ and $T$, it is trivial to construct
two GJFA over $\Gamma$ accepting $\left(\epsilon\leftarrow^{*}S\right)\backslash\left\{ \epsilon\right\} $
and $T$ respectively. On the other hand, PCP can be easily reduced
to a binary range alphabet with preserving prefix-disjointness (see,
e.g., \citep{athSAL1}). Together, we get the following:
\begin{theorem}
Given GJFA $M_{1},M_{2}$ over an $8$-letter alphabet, it is undecidable
whether $\mathrm{L}\!\left(M_{1}\right)\cap\mathrm{L}\!\left(M_{2}\right)=\emptyset$.
\end{theorem}

\section{Conclusions}

We have completed the list of closure properties of the class $\mathbf{GJFA}$
under classical language operators. Besides of that, we have pointed
out that the class $\mathbf{GJFA}$ can be put into the frameworks
of graph-controlled insertion-deletion systems and Galiukschov semicontextual
grammars, and we have referred to a former result that implies undecidability
of intersection emptiness for GJFA languages.

It is a task for future research to provide really alternative descriptions
of the class $\mathbf{GJFA}$. There also remain open questions about
decidability, specifically regarding equivalence, universality, inclusion,
and regularity of GJFA, see \citep{athMED1book}.

\bibliographystyle{C:/Users/Vojta/Desktop/SYNCHRO2/ACT/ijfcs/ws-ijfcs}
\bibliography{C:/Users/Vojta/Desktop/SYNCHRO2/bib/ruco}

\end{document}